\documentclass[a4paper,11pt]{amsart}

\usepackage[latin1]{inputenc}
\usepackage{amsthm}
\usepackage{amssymb}
\usepackage{amsfonts}
\usepackage{dsfont}
\usepackage{enumitem}
\usepackage{color}
\usepackage{datetime}
\usepackage[margin=30mm]{geometry}

\shortdate
\yyyymmdddate

\newcommand{\reals}{\mathbb{R}}

\newcommand{\complex}{\mathbb{C}}
\newcommand{\naturals}{\mathbb{N}}
\newcommand{\integers}{\mathbb{Z}}

\newcommand{\bracketb}[1]{\Big[#1\Big]}
\newcommand{\bracketc}[1]{\bigg[#1\bigg]}

\newcommand{\angles}[1]{\left\langle #1 \right\rangle}

\newcommand{\norm}[1]{\left|\left|#1\right|\right|}

\newcommand{\abs}[1]{\left|#1\right|}

\newcommand{\para}[1]{\left(#1\right)}
\newcommand{\paraa}[1]{\big(#1\big)}
\newcommand{\parab}[1]{\Big(#1\Big)}
\newcommand{\parac}[1]{\bigg(#1\bigg)}

\newcommand{\spacearound}[1]{\quad#1\quad}
\newcommand{\equivalent}{\spacearound{\Leftrightarrow}}
\renewcommand{\implies}{\spacearound{\Rightarrow}}

\newtheorem{theorem}{Theorem}[section]

\newtheorem{proposition}[theorem]{Proposition}

\theoremstyle{definition}

\theoremstyle{remark}

\numberwithin{equation}{section}
\newtheorem{conjecture}[theorem]{Conjecture}

\newcommand{\xv}{\vec{x}}

\renewcommand{\mid}{\mathds{1}}

\renewcommand{\Re}{\operatorname{Re}}
\renewcommand{\Im}{\operatorname{Im}}
\renewcommand{\d}{\partial}

\newcommand{\Xv}{\vec{X}}

\newcommand{\ket}[1]{\big|#1\rangle}
\newcommand{\bra}[1]{\langle#1\big|}

\renewcommand{\L}{\Lambda}
\newcommand{\Ld}{\Lambda^\dagger}

\newcommand{\ut}{\tilde{u}}
\newcommand{\vt}{\tilde{v}}
\newcommand{\xt}{\tilde{x}}

\newcommand{\zt}{\tilde{z}}
\newcommand{\st}{\tilde{s}}

\newcommand{\Ut}{\tilde{U}}
\newcommand{\Vt}{\tilde{V}}
\newcommand{\Wd}{W^\dagger}

\newcommand{\vol}{\operatorname{vol}}
\newcommand{\Mat}{\operatorname{Mat}}
\newcommand{\Tr}{\operatorname{Tr}}
\newcommand{\End}{\operatorname{End}}

\newcommand{\zb}{\bar{z}}

\renewcommand{\tt}{\tilde{t}}
\newcommand{\Uh}{\hat{U}}
\newcommand{\Vh}{\hat{V}}
\newcommand{\St}{\tilde{S}}

\newcommand{\tth}{\hat{\tilde{t}}}
\newcommand{\Deltah}{\hat{\Delta}}

\title{Quantum Minimal Surfaces}
\author{Joakim Arnlind, Jens Hoppe, Maxim Kontsevich}

\address[Joakim Arnlind]{Dept. of Math.\\
Link\"oping University\\
581 83 Link\"oping\\
Sweden}
\email{joakim.arnlind@liu.se}

\address[Jens Hoppe]{
  Institut des Hautes Etudes Scientifiques,
Le Bois-Marie, 35, Route de Chartres, 91440 Bures-sur-Yvette,
France 
}
\email{jhoppe@ihes.fr}

\address[Maxim Kontsevich]{
Institut des Hautes Etudes Scientifiques,
Le Bois-Marie, 35, Route de Chartres, 91440 Bures-sur-Yvette,
France }
\email{maxim@ihes.fr}

\subjclass[2000]{}

\begin{document}

\maketitle

\vspace{3mm}

\begin{center}
  \textsc{\small Abstract}\vspace{1mm}\\
  \parbox{135mm}{We discuss quantum analogues of minimal surfaces in Euclidean spaces and tori.}
\end{center}

\vspace{5mm}

\section{Introduction}

\noindent
It is well-known that minimal surfaces in $\reals^d$ can be
characterized as extremal points of the so called \emph{Schild
  functional}:
\begin{align}
  S(\vec{x})=\int_\Sigma \sum_{1\le i < j\le d}\{x_i,x_j\}^2\cdot \omega
\end{align}
where $\vec{x}:\Sigma\to \reals^d$ is a map from a surface $\Sigma$
endowed with symplectic 2-form $\omega$ to Euclidean space $\reals^d$,
and $\{\cdot,\cdot\}$ denotes the Poisson bracket on $\Sigma$. More
precisely, critical points of the Schild action are either
``degenerate'' maps with 1-dimensional image (i.e. components
$x_i, i =1,\dots,d$ of $\vec{x}$ are functionally dependent and all
Poisson brackets $\{x_i,x_j\}$ vanish identically), or the image of
$\vec{x}$ is a minimal surface in $\reals^d$ and the symplectic form
$\omega$ is proportional (with a constant factor) to the volume form
associated with the induced metric on $\Sigma$.

The Euler-Lagrange equations for the Schild action are
\begin{align}\label{eq:X.double.poisson.com}
  \sum_{i=1}^d\{x_i,\{x_i,x_j\}\} = 0\qquad j=1,\ldots,d.
\end{align}
One can call \emph{quantum minimal surface} a solution of the equation
 \begin{align}\label{eq:X.double.com}
  \sum_{i=1}^d[X_i,[X_i,X_j]] = 0\qquad j=1,\ldots,d,
\end{align}
where $X_i$ are self-adjoint operators in a Hilbert space 
-- a matrix equation which is also of interest in the context of the
bosonic BFSS \cite{h:phdthesis} and IKKT \cite{ikkt:superstring}
model, as well as for many other reasons (some of which we will
comment on in Section \ref{sec:remarks}).

In \cite{ach:nms} the classical Weierstrass representation for minimal
surfaces, utilizing the existence of isothermal parameters, was
generalized to a non-commutative one, yielding triples of Weyl algebra
elements constituting non-commutative minimal surfaces. In
\cite{ah:quantizedMinimal}, on the other hand, a quantization of the
Catenoid was written as formal power-series in $\hbar$, satisfying \eqref{eq:X.double.com}.

One can give a general procedure for constructing solutions to
\eqref{eq:X.double.com} as follows: Start with an arbitrary minimal
surface $\xv(u,v)$ (given e.g., but not necessarily, in isothermal
parametrization, $g_{ab}=\d_a\xv\cdot\d_b\xv=\sqrt{g}\delta_{ab}$);
reparametrize as
\begin{align}
  &\underset{\!\widetilde{}}{\xv}(\ut,\vt):=
    \xv\paraa{u(\ut,\vt),v(\ut,\vt)}\label{eq:orig.xvutvt}
\end{align}
with the new coordinates chosen such that
\begin{align}\label{eq:J}
  &J:=\abs{\frac{\d(\ut,\vt)}{\d(u,v)}}=\sqrt{g},\qquad
\end{align}
i.e. $\{\tilde{u},\tilde{v}\}=1$. With $[\Ut,\Vt]=i\hbar\mid$,
\begin{align}
 X_i:=\xt_i(\Ut,\Vt) 
\end{align}
will then satisfy \eqref{eq:X.double.com}, to
lowest order in $\hbar$, due to
\begin{align}
  \frac{1}{\sqrt{g}}\d_a\paraa{\sqrt{g}g^{ab}\d_b\xv}\equiv\{x_i\{x_i,\xv\}\}
\end{align}
when
\begin{align}
  \{f,h\}:=\frac{1}{\sqrt{g}}\epsilon^{ab}(\d_af)(\d_bh);
\end{align}
note that \eqref{eq:orig.xvutvt} implies $\sqrt{\tilde{g}}=1$,
i.e. furnishing a transformation to a parametrization where the
determinant of the first fundamental form is constant ($=1$).

Let us first consider the following (non-compact) examples: Catenoids
(for which, due to the rotational symmetry, one can easily obtain
rather explicit expressions, resp. existence proofs to all orders) and
Enneper surfaces (where, despite of, again, fairly explicit formulas
one sees much of the difficulties involved concerning the general
case).

Note that for $d=3$ equation \eqref{eq:X.double.com}, with
$W:=X_1+iX_2$ and $Z:=X_3$, reads 
\begin{equation}\label{eq:Deltah.W.Z}
  \begin{split}
    &\Deltah(W) = \frac{1}{2}[W,[\Wd,W]]+[Z,[Z,W]] = 0\\
    &\Deltah(Z) = \frac{1}{2}[W,[\Wd,Z]] + \frac{1}{2}[\Wd,[W,Z]] = 0.
  \end{split}
\end{equation}

\section{Catenoid}

\noindent
Let us consider the catenoids
\begin{align}
  &\xv_{\pm} =
    \paraa{a\cosh(v-v_0)\cos(u-u_0),
  a\cosh(v-v_0)\sin(u-u_0), \pm av}+\xv_0  \\
  &(g_{ab})=a^2\cosh^2(v-v_0)
  \begin{pmatrix}
    1 & 0 \\ 0 & 1
  \end{pmatrix}
                 =\sqrt{g}\mid.\notag
\end{align}
Reparametrizing $\xv_{\pm}$ as $\xv_{\pm}(\ut=u,\vt)$ in accordance
with \eqref{eq:J}, i.e.
\begin{align}\label{eq:cat.dvtdv}
  \frac{d\vt}{dv} = a^2\cosh^2(v-v_0)\qquad
  \vt(v) = \frac{a^2}{2}\parab{v-v_0+\frac{1}{2}\sinh 2(v-v_0)}+\vt_0
\end{align}
gives
\begin{align*}
  \Xv_{\pm} =
  \parab{a&\cosh\paraa{v(\Vt-\tilde{v}_0\mid)-v_0\mid}\cos(\Ut-u_0\mid),\\
  &a\cosh\paraa{v(\Vt-\tilde{v}_0\mid)-v_0\mid}\sin(\Ut-u_0\mid),
  \pm a v(\Vt-\tilde{v}_0\mid))}+\xv_0
\end{align*}
with a natural representation on the basis $\ket{n}\hat{=}e^{-in\varphi}$ as
\begin{align}
  e^{i(\Ut-u_0\mid)}\ket{n} = e^{i\varphi}\ket{n}=\ket{n-1}\qquad
  (\Vt-\tilde{v}_0\mid)\ket{n} = -i\hbar\frac{\d}{\d\varphi}\ket{n}
  =-\hbar n\ket{n}
\end{align}
which leads to the Ansatz
\begin{align}\label{eq:cat.WWZpm}
  W\ket{n} = w_n\ket{n-1}\qquad
  Z^{(\pm)}\ket{n} =z_n\ket{n}.
\end{align}
Inserting \eqref{eq:cat.WWZpm} into \eqref{eq:Deltah.W.Z} one obtains
the recursion relations
\begin{align}
  &2(z_n-z_{n-1})^2 = r_{n+1}+r_{n-1}-2r_n\label{eq:cat.znrn}\\
  &r_n(z_n-z_{n-1}) = r_{n+1}(z_{n+1}-z_n)=\textrm{const}=c\label{eq:cat.znrnc}
\end{align}
i.e. $r_n:=|w_n|^2$ determined by
\begin{align}\label{eq:cat.rndiffc}
  (r_{n+1}-r_n)-(r_n-r_{n-1}) = \frac{2c^2}{r_n^2}
\end{align}
and then $z_n$ given via \eqref{eq:cat.znrnc},
\begin{align}
  z_n = \frac{c}{r_n}+z_{n-1}.\label{eq:cat.zn}
\end{align}
Denoting $v-v_0$ by $q$ and $\vt-\vt_0$ by $p$
(i.e. $p=a^2/2(q+\frac{1}{2}\sinh 2q)$) being an odd function of $q$,
resp. $q$ an odd function of $p$, and $\frac{dp}{dq}=a^2\cosh^2q$
\eqref{eq:cat.rndiffc} and \eqref{eq:cat.zn} are by construction
solved in the $\hbar\to 0$ limit ($n$ arbitrary but fixed) by
\begin{equation}
  \begin{split}
    &r_n = a^2\cosh^2\paraa{q(-\hbar n)}\\
    &z_n = \pm a\paraa{q(-\hbar n)+v_0}+(\xv_0)_3
  \end{split}
\end{equation}
which can easily be verified, as \eqref{eq:cat.znrn} and
\eqref{eq:cat.znrnc} in this limit become the differential equations
\begin{align}\label{eq:cat.zprzpD}
  2z'^2=r''\qquad rz'=D\quad(D=\lim_{\hbar\to 0}c/\hbar),
\end{align}
which \emph{are} satisfied (using $\frac{dq}{dp}=1/(a^2\cosh^2q(p))$) by
\begin{equation}\label{eq:cat.rpzp}
  \begin{split}
    &r(p) = a^2\cosh^2q(p)\\
    &z(p) = \pm a\paraa{q(p)+v_0}+(\xv_0)_3\qquad D=\pm a.
  \end{split}
\end{equation}
Let us now consider the general solution to the recursion relations
\eqref{eq:cat.znrn} and \eqref{eq:cat.znrnc}. First of all, since
$r_n=|w_n|^2$ one is interested in \emph{positive} solutions,
i.e. solutions $(\{r_n\}_{n\in\integers},\{z_n\}_{n\in\integers})$
with $r_n>0$ for $n\in\integers$. Moreover, we are interested in
\emph{non-constant} solutions (noting that \eqref{eq:cat.znrn} and
\eqref{eq:cat.znrnc} have constant solutions), and it is easy to see that
for a non-constant solution, both sequences $\{r_n\}$ and $\{z_n\}$ are
necessarily non-constant. Let us start by showing that for appropriate
initial conditions, one obtains a positive and non-constant solution
to the recursion relations.

\begin{proposition}\label{prop:rz.recursion}
  For $c,r_0,r_1,z_0\in\reals$ such that $c\neq 0$ and
  \begin{align*}
    0<r_0\leq r_1\leq r_0+\frac{2c^2}{r_0^2}
  \end{align*}
  there exists a non-constant solution to the system
  \begin{align}
    &2(z_n-z_{n-1})^2 = r_{n+1}+r_{n-1}-2r_{n}\label{eq:DeltaWn.zero}\\
    &r_n(z_n-z_{n-1}) = r_{n+1}(z_{n+1}-z_n)\label{eq:DeltaZn.zero}
  \end{align}
  such that
  \begin{alignat*}{2}
    &r_{n+1}\geq r_{n}>0 &\qquad &\text{for }n\geq 0\\
    &r_{n-1}\geq r_{n}>0 & &\text{for }n\leq 0.
  \end{alignat*}
  The solution is given recursively by
  \begin{alignat}{2}
    r_{n} &= 2r_{n-1}-r_{n-2}+\frac{2c^2}{r_{n-1}^2} &\qquad &\text{for }n\geq 2\label{eq:rposn.rec}\\
    r_{n} &= 2r_{n+1}-r_{n+2}+\frac{2c^2}{r_{n+1}^2} & &\text{for }n\leq -1\label{eq:rnegn.rec}\\
    z_{n} &= z_{n-1}+\frac{c}{r_n} & &\text{for }n\geq 1\label{eq:zposn.rec}\\
    z_{n} &= z_{n+1}-\frac{c}{r_{n+1}} & &\text{for }n\leq -1.\label{eq:znegn.rec}
  \end{alignat}
\end{proposition}

\begin{proof}
  First, let us note that \eqref{eq:DeltaZn.zero} implies that
  $c=r_n(z_n-z_{n-1})$ is independent of $n$. If $c\neq 0$ then
  $r_n\neq 0$ for $n\in\integers$, which implies that
  \begin{align}
    z_n-z_{n-1} = \frac{c}{r_n},\label{eq:znc}
  \end{align}
  yielding \eqref{eq:zposn.rec} and \eqref{eq:znegn.rec}.  Inserting
  \eqref{eq:znc} into \eqref{eq:DeltaWn.zero} gives
  \eqref{eq:rposn.rec} and \eqref{eq:rnegn.rec}. Conversely, if
  $r_n(z_n-z_{n-1})=c$ it is clear that
  \eqref{eq:rposn.rec}--\eqref{eq:znegn.rec} satisfy
  \eqref{eq:DeltaWn.zero} and \eqref{eq:DeltaZn.zero}.

  Now, assume that $n\geq 1$ and that $r_n\geq r_{n-1}>0$, which is
  true by assumption when $n=1$; let us show that if we define
  $r_{n+1}$ as in \eqref{eq:rposn.rec} then $r_{n+1}\geq
  r_{n}>0$. Equation \eqref{eq:rposn.rec} gives
  \begin{align*}
    r_{n+1}-r_n = r_n-r_{n-1}+\frac{2c^2}{r_{n-1}^2}>0
  \end{align*}
  since $r_n-r_{n-1}\geq 0$ and $c\neq 0$. Moreover, it is clear that
  $r_{n+1}>0$ since $r_n>0$. Thus, we have shown that
  $r_{n+1}\geq r_{n}>0$, and it follows by induction that this is true
  for all $n\geq 0$. For $n\leq -1$ we start by defining
  \begin{align*}
    r_{-1} = 2r_0-r_1+\frac{2c^2}{r_0^2}
  \end{align*}
  and note that
  \begin{align*}
    r_{-1}-r_0=r_0-r_1+\frac{2c^2}{r_0^2}\geq 0
  \end{align*}
  by using the assumption, and furthermore that $r_{-1}>0$ since
  $r_0>0$.  Next, for $n\leq -1$ one assumes that
  $r_{n}\geq r_{n+1}> 0$ and defines $r_{n-1}$ via
  \eqref{eq:rnegn.rec}, implying
  \begin{align*}
    r_{n-1}-r_n=r_n-r_{n+1}+\frac{2c^2}{r_{n}^2}\geq 0
  \end{align*}
  since $r_n-r_{n+1}\geq 0$. Moreover, it is clear that $r_{n-1}>0$
  since $r_n>0$.  Hence, it follows by induction that
  $r_n\geq r_{n+1}>0$ for all $n\leq -1$.
\end{proof}

\noindent Note that, for any solution to the recursion relations, one
may always shift $n$ to obtain another solution with a minimum value
at $r_{n_0}$ for arbitrary $n_0\in\integers$. The next result shows
that these are indeed all possible non-constant solutions.

\begin{proposition}
  Let $(\{z_n\}_{n\in\integers},\{r_n\}_{n\in\integers})$ be a
  positive non-constant solution of \eqref{eq:DeltaWn.zero} and
  \eqref{eq:DeltaZn.zero}. Then there exists $n_0\in\integers$,
  $c\in\reals\backslash\{0\}$ and $\delta\in(-1,1]$ such that
  \begin{enumerate}
  \item $r_{n_0}=\min_{n\in\integers}r_n$,
  \item $\displaystyle r_{n_0+1}=r_{n_0}+(1-\delta)\frac{c^2}{r_{n_0}^2}$,
  \item $r_{n+1}>r_{n}$ for $n\geq n_0+1$,
  \item $r_{n-1}>r_{n}$ for $n\leq n_0$. 
  \end{enumerate}
  Moreover, $\{z_n\}_{n\in\integers}$ is strictly increasing
  if $c>0$ and strictly decreasing if $c<0$.  
\end{proposition}

\begin{proof}
  First, let us note that \eqref{eq:DeltaZn.zero} implies that
  there exists $c\in\reals$ such that
  \begin{align*}
    r_n(z_n-z_{n-1})=c
  \end{align*}
  for all $n\in\integers$. Moreover, if $z_n-z_{n-1}=0$ for some
  $n\in\integers$ then $c=0$, which implies that $z_n-z_{n-1}=0$ for
  all $n$ (by using that $r_n>0)$. Thus, since the solution is assumed
  to be non-constant, it follows that $c\neq 0$. Equations
  \eqref{eq:DeltaWn.zero} and \eqref{eq:DeltaZn.zero} can then be
  written as
  \begin{align}
    &r_{n+1}-r_n=r_n-r_{n-1} + \frac{2c^2}{r_n^2}\label{eq:rndiff}\\
    &z_n-z_{n-1} = \frac{c}{r_n}\label{eq:zndiffc}
  \end{align}
  since $r_n>0$. It is clear from \eqref{eq:rndiff} that if
  $r_N-r_{N-1}\geq 0$ for some $N\in\integers$, then $r_{n+1}>r_n$
  for all $n\geq N$. Similarly, if $r_{N+1}-r_N\leq 0$ then
  $r_{n-1}>r_n$ for all $n\leq N$.
  
  Since the solution is non-constant, there exists $N\in\integers$
  such that $r_{N}-r_{N-1}\neq 0$; assume that $r_{N}-r_{N-1}<0$ (the
  argument for $r_N-r_{N-1}>0$ is completely analogous). It follows
  from \eqref{eq:rndiff} that
  \begin{align*}
    r_N-r_{N-1}<r_N-r_{N-1}+\frac{2c^2}{r_N^2}=r_{N+1}-r_{N}
  \end{align*}
  and, by induction
  \begin{align*}
    r_N-r_{N-1}<r_{N+1}-r_{N}<r_{N+2}-r_{N+1}<\cdots
  \end{align*}
  implying that there exists $n_0\in\integers$ such that
  $r_{n_0}-r_{n_0-1}<0$ and $r_{n_0+1}-r_{n_0}\geq 0$ (since
  $2c^2/r_n^2$ is bounded from below by $2c^2/r_N^2>0$). By the
  previous argument one may also conclude that $r_{n-1}>r_{n}$ for
  $n\leq n_0$ and $r_{n+1}>r_n$ for $n\geq n_0+1$. Moreover, it
  follows that $r_{n_0}=\min_{n\in\integers}r_n$. For $n=n_0$,
  \eqref{eq:rndiff} gives
  \begin{align*}
    r_{n_0+1}-r_{n_0}=r_{n_0}-r_{n_0-1}+\frac{2c^2}{r_{n_0}^2}
  \end{align*}
  implying that $r_{n_0}<r_{n_0}+2c^2/r_{n_0}^2$ since
  $r_{n_0}-r_{n_0-1}<0$, which, together with
  $r_{n_0+1}-r_{n_0}\geq 0$ gives
  \begin{align*}
    r_{n_0}\leq r_{n_0+1}< r_{n_0}+\frac{2c^2}{r_{n_0}^2}. 
  \end{align*}
  Finally, it follows directly from \eqref{eq:zndiffc} that
  $\{z_n\}$ is strictly increasing if $c>0$ and strictly decreasing
  if $c<0$.
\end{proof}
\noindent
What about the \emph{general} solution of \eqref{eq:cat.zprzpD}:
\begin{align}
  &r''r^2=2D^2\implies r'^2=-\frac{4D^2}{r}+\alpha\implies\\
  &p = \int\frac{dr}{\sqrt{\alpha-\frac{4D^2}{r}}}\underset{r=\frac{4D^2}{\alpha}\cosh^2q}{\overset{q>0}{=}}
    \frac{8D^2}{\alpha^{3/2}}\int\cosh^2qdq\implies\\
  &p-p_0 = \frac{8D^2}{2\alpha^{3/2}}\parab{q+\frac{1}{2}\sinh 2q}\qquad
    \frac{dp}{dq} = \frac{8D^2}{\alpha^{3/2}}\cosh^2 q\notag\\
  &z' = \frac{D}{r} = \frac{\alpha}{4D}\frac{1}{\cosh^2q}=\frac{\alpha}{4D}\frac{dq}{dp}\frac{8D^2}{\alpha^{3/2}}\implies
    z=\frac{2D}{\sqrt{\alpha}}q(p)+z_0\\
  &(x_1+ix_2)e^{-i\varphi} = \frac{|2D|}{\sqrt{\alpha}}\cosh q(p),\notag
\end{align}
which is a catenoid (with $g(p,u)=\alpha/4$); with $\frac{2D}{\sqrt{\alpha}}=\pm a$ one gets
\eqref{eq:cat.rpzp}, but the relation between $p$ and $q$ having an
extra factor of $\frac{2}{\sqrt{\alpha}}$, i.e.
\begin{align}
  \frac{dp}{dq}=\frac{2}{\sqrt{\alpha}}\sqrt{g}(q) 
\end{align}
instead of $\sqrt{g}$. Note that the discretization of
\eqref{eq:cat.rpzp} satisfies the recursion relations
\eqref{eq:cat.znrn} and \eqref{eq:cat.znrnc} also for $n\to\infty$
(fixed $\hbar$), which can be seen as follows: \eqref{eq:cat.dvtdv}
implies that for large $p=\vt-\vt_0$
\begin{align}
  q=v-v_0\approx\frac{1}{2}\ln p+\frac{1}{2}\ln\frac{8}{a^2}
  -a^2\frac{\ln\para{\frac{8p}{a^2}}}{8p}+\cdots,
\end{align}
and therefore, for large $|n|$
\begin{align}
  r_n=a^2\cosh^2\paraa{q(-\hbar n)}\sim 2\hbar|n|-\frac{a^2}{2}\ln|n|,
\end{align}
indeed satisfying \eqref{eq:cat.rndiffc} to leading order ($c^2=a^2\hbar^2$).

\section{Enneper surfaces}

\noindent
For Enneper-type surfaces the data entering the Weierstrass
representation
\begin{align}
  \xv = \Re\int\vec{\varphi}(z)dz
\end{align}
are
\begin{align}
  \vec{\varphi} = \paraa{(1-z^2)z^{N-2},i(1+z^2)z^{N-2},2z^{N-1}}
  \qquad(N\geq 2)
\end{align}
the simplest case $(N=2)$ giving
\begin{equation}
  \begin{split}
    &\xt(u,v) = \paraa{u-\tfrac{1}{3}u^3+uv^2,
      \tfrac{1}{3}v^3-v-vu^2,u^2-v^2}\\
    &(g_{ab}) = (1+u^2+v^2)^2
    \begin{pmatrix}
      1 & 0 \\ 0 & 1
    \end{pmatrix},
  \end{split}
\end{equation}
while more generally, with $s:=|z|^2=u^2+v^2$,
\begin{align}
  \sqrt{g} = \frac{1}{2}\vec{\varphi}\vec{\varphi}^\ast =
  s^{N-2}(1+s)^2.
\end{align}
Trying to solve \eqref{eq:J} with the Ansatz
\begin{equation}
  \begin{split}
    &\ut = uh(u^2+v^2)\quad\paraa{=u\sqrt{\tfrac{\tilde{s}}{s}}}\\
    &\vt = vh(u^2+v^2)\quad\paraa{=v\sqrt{\tfrac{\tilde{s}}{s}}}
  \end{split}
\end{equation}
gives
\begin{align}
  \frac{d\tilde{s}}{ds} =s^{N-2}(1+s)^2=s^N+2s^{N-1}+s^{N-2}
\end{align}
Let us now restrict to $N=2$, giving
\begin{align}
  &\tilde{s} = \frac{1}{3}s^3+s^2+s+c
  =\frac{1}{3}(1+s)^3+c-\frac{1}{3}
\end{align}
which can be easily inverted to give
\begin{align}
  s = (3\tilde{s}+1-3c)^{1/3}-1,
\end{align}
hence yielding expressions for the inverse transformation
\begin{align}\label{eq:uv.utvt}
  &u = \ut\sqrt{\frac{s}{\tilde{s}}}=\ut\tilde{h}(\tilde{s})\qquad
  v =\vt\sqrt{\frac{s}{\tilde{s}}} =\vt\tilde{h}(\tilde{s})
\end{align}
which is needed to obtain the $X_i$. Defining $w=x_1+ix_2$ and $z=u+iv$, resp.
\begin{align*}
  W = X_1+iX_2\qquad\Lambda = U+iV
\end{align*}
one finds for this (Enneper) case
\begin{align}
  \{w,\bar{w}\}=\frac{2}{i}\parac{1-\frac{2}{1+z\bar{z}}}\qquad
  \{x_3,w\} = \frac{2}{i}\frac{z}{1+z\bar{z}}
\end{align}
satisfying
\begin{equation}
  \begin{split}
    &\frac{1}{2}\{\{w,\bar{w}\},w\}=\{\{x_3,w\},x_3\}\\
    &\{\{x_3,w\},\bar{w}\}+\{\{x_3,\bar{w}\},w\} = 0
  \end{split}
\end{equation}
while \eqref{eq:X.double.com} resp.
\begin{equation}
  \begin{split}
    &\frac{1}{2}[W,[W,\Wd]]=[X_3,[X_3,W]]\\
    &[[X_3,W],\Wd]+[[X_3,\Wd],W] = 0
  \end{split}
\end{equation}
should (in leading order) be solved by
\begin{align}
  W = \Ld-\frac{1}{3}\L^3\qquad X_3=\frac{1}{2}\paraa{\L^2+{\L^\dagger}^2}
\end{align}
with
\begin{align}
  \L = f(\hat{N})Tf(\hat{N})
\end{align}
where
\begin{align}
  T\ket{n} = \ket{n-1}\quad\text{ and }\quad
  \hat{N}\ket{n}=n\ket{n}\qquad(n\in\naturals_0)
\end{align}
is the number operator. Due to (cp. \eqref{eq:uv.utvt}) 
\begin{align}
  &z=\tilde{z}\tilde{h}(\tilde{s})
    =:f(\tilde{s})\frac{\tilde{z}}{\sqrt{\tilde{s}}}f(\tilde{s})\\
  &f(\tilde{s}) = \paraa{(3\tilde{s}+1-3c)^{1/3}-1}^{\frac{1}{4}}.
\end{align}


\noindent
Analogously, 
\begin{align*}
    &z = \frac{\zt}{\sqrt{\st}}\sqrt{\paraa{3(\st-c)+1}^{1/3}-1}
\end{align*}
suggests
\begin{align}
  &Z\ket{n} = \sqrt{\paraa{6\hbar n+3\hbar+1-3c}^{1/3}-1}\sqrt{\frac{2\hbar n}{2\hbar n+\hbar}}\ket{n-1}=\Lambda\ket{n}=\lambda_n\ket{n-1}
\end{align}
due to
\begin{equation}
  \begin{split}
    &\St\ket{n} = 2\hbar\paraa{n+\tfrac{1}{2}}\ket{n}\\
    &\Ut+i\Vt=\sqrt{2\hbar}a=\sqrt{\hbar}\paraa{\d_x+x}\\
    &a\ket{n} = \sqrt{n}\ket{n-1}\qquad
    a^\dagger\ket{n} = \sqrt{n+1}\ket{n+1}\qquad
    a^\dagger a\ket{n} = n\ket{n}.
  \end{split}
\end{equation}
\begin{align}      
  &|\lambda_n|^2
    =\bracketb{\paraa{6\hbar n+3\hbar+1-3c}^{1/3}-1}\frac{2n}{2n+1}
    \qquad n=0,1,2,\ldots
\end{align}
satisfies to leading (and sub-leading if $n\to\infty$, $\hbar$ fixed) order the recursion relations
\begin{align}
  &(\sigma_{n+1}-\sigma_n)\paraa{2+\sigma_n+\sigma_{n+1}}^2=8\hbar
  \equivalent [\Lambda,\Ld] = \frac{2\hbar}{\paraa{1+\tfrac{1}{2}(\L\Ld+\Ld\L)}^2}\label{eq:enas.sigma.recursion}\\
  &\sigma_0=0\quad\sigma_1=2\hbar-2\hbar^2+\cdots\implies c=0.\notag
\end{align}
Inserting 
\begin{align*}
  X_1+iX_2=W=\Ld-\frac{1}{3}\L^3\qquad X_3=\frac{1}{2}\parab{\L^2+{\Ld}^2}
\end{align*}
into $[X_i,[X_i,X_j]]=0$ gives
\begin{align}
  &[\L,[\L,\Ld]-\tfrac{1}{9}[\L^3,{\Ld}^3]]
  +[\L,\tfrac{1}{2}[\L^2,{\Ld}^2]+\tfrac{1}{6}[\L^2,{\Ld}^3]]=0\\
  &[\L^2,[\L,\Ld]+\tfrac{1}{9}[\L^3,{\Ld}^3]]
    +\frac{2}{3}[\L^3,[\L,{\Ld}^2]]=0\\
  &\frac{1}{3}[\L^3,\tfrac{1}{9}[\L^3,{\Ld}^3]+\tfrac{1}{2}[\L^2,{\Ld}^2]-[\L,\Ld]]
    +\frac{1}{2}[\L^2,[\L^2,\Ld]]=0.
\end{align}
In the classical limit the above equations are satisfied, using
\begin{align}\label{eq:enas.zzb}
 \{z,\zb\} =-\frac{2i}{(1+z\zb)^2},
\end{align}
as reducible to $\{z,1\}=0$, $\{z^2,1\}=0$,
$\{z^3,1\}=0$. Q-analogue of \eqref{eq:enas.zzb}
\begin{align}
  [\L,\Ld] = \frac{2\hbar}{\paraa{1+\tfrac{1}{2}(\L\Ld+\Ld\L)}^2}
\end{align}
has solutions of the form
\begin{align*}
  \L =
  \begin{pmatrix}
    0 & \lambda_1\\
    0 & 0 & \lambda_2\\
    & &  \ddots
  \end{pmatrix}
\end{align*}
with $\sigma_n=|\lambda_n|^2$ satisfying
\begin{align*}
  (\sigma_{n+1}-\sigma_n)\paraa{2+\sigma_{n+1}+\sigma_n}^2=8\hbar
\end{align*}
with unique solution $\sigma_0=0<\sigma_1<\sigma_2<\cdots$.

\section{Helicoid}

\noindent
Parametrizing the helicoid in $\reals^3$ as
\begin{align*}
  \xv(u,v) = (\sinh(v)\cos(u),\sinh(v)\sin(u),u)
\end{align*}
gives $(g_{ab})=\cosh^2(u)\mid$ and $\sqrt{g}=\cosh^2(u)$. A solution to
\begin{align*}
  \abs{\frac{\d(\ut,\vt)}{\d(u,v)}} = \sqrt{g}=\cosh^2(u)
\end{align*}
is again given by $\ut=u$ and $\vt=v/2+\sinh(2v)/4$, which implies
that
\begin{align*}
  &w = x^1+ix^2 = \sinh(v(\vt))e^{i\ut}\\
  &z = x^3 = \ut.
\end{align*}
For the helicoid, choosing a representation of $\Ut$ and $\Vt$ on
smooth functions as
\begin{align*}
  &(\Ut f)(x) = i\hbar f'(x)\\
  &(\Vt f)(x) = xf(x),
\end{align*}
one may interpret $e^{i\Ut}$ as a shift
operator. Defining the following operators
\begin{alignat*}{2}
  &(Wf)(x) = w(x)f(x-\hbar) &\qquad
  &(W^\dagger f)(x)=\overline{w(x+\hbar)}f(x+\hbar)\\
  &(Zf)(x) = i\hbar f'(x) &
  &(Z^\dagger f)(x)=i\hbar f'(x),
\end{alignat*}
it follows that $\Delta(W)=\Delta(Z)=0$ is equivalent to the single equation
\begin{align}\label{eq:helicoid.equation}
  w(x)\paraa{2|w(x)|^2-|w(x-\hbar)|^2-|w(x+\hbar)|^2} = 2\hbar^2w''(x).
\end{align}
If $w(x-\hbar)\neq 0$ one may write the above equation as
\begin{align}\label{eq:helicoid.eq.2hbar}
  |w(x)|^2 = 2|w(x-\hbar)|^2-|w(x-2\hbar)^2|-2\hbar^2\frac{w''(x-\hbar)}{w(x-\hbar)}.
\end{align}

\section{Complex hyperbola}\label{sec:hyperbola}

\noindent
In contrast to the previous examples, let us now consider a surface in
$\reals^4$. Parametrizing $(x_1+ix_2)(x_3+ix_4)=\alpha\in\reals$
($\Leftrightarrow$ $x_1x_3-x_2x_4=\alpha$, $x_1x_4+x_2x_3=0$)
as
$\xv(t,u)=(t\cos u,t\sin u,\tfrac{\alpha}{t}\cos u,-\tfrac{\alpha}{t}\sin
u)_{t>0}$, one gets
\begin{align*}
  (g_{ab}) = \para{1+\tfrac{\alpha^2}{t^4}}
  \begin{pmatrix}
    1 & 0 \\ 0 & t^2
  \end{pmatrix}\qquad
                 \sqrt{g}=t\para{1+\tfrac{\alpha^2}{t^4}} = t+\alpha^2t^{-3}
\end{align*}
Using the general strategy, namely reparametrizing
$\xv(t,u)=\tilde{\xv}(\tt,u)$ with
\begin{align}
  \tfrac{d\tt}{dt}=\sqrt{g}=t+\alpha^2t^{-3},
\end{align}
i.e. $\tt=\frac{1}{2}t^2-\frac{\alpha^2}{2t^2}+\gamma$, respectively
\begin{align}
  t(\tt)=\sqrt{\tt-\gamma+\sqrt{(\tt-\gamma)^2+\alpha^2}},
\end{align}
$\xt_1+i\xt_2=t(\tt)e^{iu}=e^{iu}t(\tt)$,
$\xt_3+i\xt_4=\frac{1}{t(\tt)}e^{-iu}$, leading to (in the
representation $\Uh=\varphi$, $\Vh=\tth=-i\hbar\frac{\d}{\d\varphi}$
acting on $e^{-in\varphi}\hat{=}\ket{n}$)
\begin{align}
  W\ket{n} = (X_1+iX_2)\ket{n} = w_n\ket{n-1}\qquad
  (X_3+iX_4)\ket{n} = \frac{\alpha}{w_{n+1}}\ket{n+1}
\end{align}
with $w_n=t(-\hbar n)$, and 
\begin{align}\label{eq:rnwngamma}
  r_n=|w_n|^2=-\hbar n-\gamma+\sqrt{(-\hbar n-\gamma)^2+\alpha^2} 
\end{align}
giving exact solutions to \eqref{eq:X.double.com}; we will come back
to this example in Section~\ref{sec:quantized.complex.curves}.

\section{Remarks and conjectures}\label{sec:remarks}

\subsection{I -- Quantization}

\begin{center}
  \begin{tabular}{lcl}
    \parbox[t]{55mm}{$(M,\omega)$ compact $C^\infty$ symplectic manifold of dimension $2d$} & $\rightsquigarrow$ & \parbox[t]{7cm}{Hilbert space $\mathcal{H}$\\ $\dim\mathcal{H}=\frac{1}{(2\pi\hbar)^d}\int_M\frac{\omega^d}{d!}\cdot(1+O(\hbar))$}\\
    $f\in C^\infty(M,\reals)$ & $\rightsquigarrow$ & \parbox[t]{7cm}{self-adjoint $\hat{f}\in\End(\mathcal{H})$\\ $\hat{f}\cdot\hat{g}=\widehat{fg}+\frac{i\hbar}{2}\widehat{\{f,g\}}+O(\hbar^2)$}\\
    $\varphi\in C^\infty(M,\complex)$, $|\varphi|=1$ & $\rightsquigarrow$ &
                                                                            \parbox[t]{7cm}{unitary operator $\hat{\varphi}\in U(\mathcal{H})$\\
    $\hat{\varphi}_1\cdot\hat{\varphi}_2=\widehat{\varphi_1\varphi_2}\cdot e^{\frac{i\hbar}{2}\{\log\varphi_1,\log\varphi_2\}+O(\hbar^2)}$\\(well-defined locally on $M$)}
  \end{tabular}
\end{center}
\begin{align*}
  &i\hbar\{\cdot,\cdot\}\simeq [\cdot,\cdot]\\
  &\frac{1}{2\pi\hbar}\int_M\cdot\quad\simeq \Tr(\cdot)\qquad\text{if $d=1$.}
\end{align*}
Example: $M=U(1)\times U(1)$, $\varphi_1=e^{i\theta_1}$,
$\varphi_2=e^{i\theta_2}$ with
$\theta_1,\theta_2\in \reals/2\pi\integers$, and
$\omega=(d\theta_1\wedge d\theta_2)/2\pi$, $\hbar=1/N$:
\begin{align*}
  \hat{\varphi}_1 =
  \begin{pmatrix}
    0 & 0 & \cdots & 0 & 1\\
    1 & 0 & \cdots & 0 & 0\\
    0 & 1 & \cdots & 0 & 0\\
    \vdots & \vdots &\ddots &\vdots &\vdots\\
    0 & 0 & \cdots & 1 & 0
  \end{pmatrix}\qquad\quad
                         \hat{\varphi}_2=
                         \begin{pmatrix}
                           1 & & \\
                           & e^{2\pi i/N} &\\
                           & & \ddots &\\
                           & & & e^{2\pi i(N-1)/N}
                         \end{pmatrix}.
\end{align*}
\subsection{II} Map $\varphi:M\to\reals^n$ given by $x_1,\ldots,x_n\in C^\infty(M,\reals)$. Functional
\begin{align*}
  S_d(\varphi) = \vol(M)\int_M\sum_{i<j}\{x_i,x_j\}^2
  =\int_M 1\cdot\int_M\sum_{i<j}\{x_i,x_j\}^2
\end{align*}
invariant under rescaling of symplectic form $\omega\to\lambda\omega$
for $\lambda\in\reals_{>0}$. The quantum version is given by
$X_1,\ldots,X_n\in\Mat(N\times N,\complex)$ with $X_i^\dagger = X_i$
(think as $X_i=\hat{x}_i$) and
\begin{align*}
  S_q = 2\pi\hbar\Tr\mid_{N\times N}\cdot
  2\pi\hbar\Tr\sum_{i<j}\para{\frac{1}{i\hbar}[X_i,X_j]}^2
  =-(2\pi)^2\cdot N\cdot \Tr\sum_{i<j}[X_i,X_j]^2.
\end{align*}
Put constraints to make this compact: e.g. sphere $\sum x_i^2=1$,
ellipsoid $\sum(x_i/r_i)^2=1$ (where $r_1,\ldots,r_n>0$). Quantum
sphere: $\sum X_i^2=\mid_{N\times N}$, quantum ellipsoid
$\sum X_i^2/r_i^2=\mid_{N\times N}$.

\begin{conjecture}
  Critical values of $S_q$ with given constraint $\approx$ critical
  values of $S_d$ as $N\to+\infty$ (these are essentially squares of  volumes of
  minimal surfaces).
\end{conjecture}

\noindent
Note that $d>1$ does not survive, wrong scaling.

\subsection{III}

\underline{Version}: target = flat torus $U(1)\times\cdots\times
U(1)=(S^1)^n$. Map $M\to(S^1)^n$ is given by
$\varphi_1,\ldots,\varphi_n\in C^\infty(M,\complex)$ with
$|\varphi_i|^2=1$.
\begin{align*}
  S_d = \int_M\sum_{i<j}\{\log\varphi_i,\log\varphi_j\}^2\cdot\int_M 1
\end{align*}
Critical points = minimal surfaces in $(S^1)^n$ parametrized by
symplectic surface. $\omega$ = const $\cdot$ Riemannian volume for the
induced metric.

\underline{Weierstrass parametrization}: M is a complex curve $C$ of genus
$g\geq 1$, $\alpha_1,\ldots,\alpha_n\in\Omega^1(C)=\Gamma(C,T^\ast_C)$
are holomorphic $1$-forms and $i\Im\alpha_j=d\varphi_j$. Constraints:
\begin{enumerate}
\item $\displaystyle\sum_{i=1}^n\alpha_i^2=0\in\Gamma(C,(T^\ast_C)^{\otimes 2})$
  (quadratic differentials, space of $\dim_\complex = 3g-3$).
\item $[\Im\alpha_i]\in H^1(C,2\pi\integers)$ for all $i$
\end{enumerate}

$\rightsquigarrow$ discrete countable subset of $\reals_{\geq 0}$ of
critical values (depending on $n$).

$\Phi_1,\ldots,\Phi_n\in U(N)$ ($N\to+\infty$ think as $\Phi_i=\hat{\varphi}_i$). Define
\begin{align*}
  S_q = N\cdot\sum_{i<j}\Tr\paraa{2\cdot\mid_{N\times N}-\Phi_i\Phi_j\Phi_i^{-1}\Phi_j^{-1}-\Phi_j\Phi_i\Phi_j^{-1}\Phi_i^{-1}}
\end{align*}
(see \cite{dn:noncommutative.field.theory} for related considerations,
including the lattice twisted Eguchi-Kawai model). Reasoning: If
$\Phi\in U(N)$ is close to $\mid_{N\times N}$ then
$-\Tr(\log\Phi)^2\sim\Tr(2-\Phi-\Phi^{-1})$ ($\Phi=e^{i\theta}$,
$\theta^2\sim 2-2\cos\theta=2-e^{i\theta}-e^{-i\theta}$).

``Equations of motion'': Critical points $\frac{\delta S_q}{\delta\Phi_i}=0$ gives
\begin{align*}
  \sum_{i\neq j}\parab{\Phi_j\Phi_i^{-1}\Phi_j^{-1}-\Phi_i^{-1}\Phi_j^{-1}\Phi_i\Phi_j^{-1}\Phi_i^{-1} + \Phi_j^{-1}\Phi_i^{-1}\Phi_j - \Phi_i^{-1}\Phi_j\Phi_i\Phi_j\Phi_i^{-1}} = 0
\end{align*}
for all $i$. Multiply by $\Phi_i$ from the left:
\begin{align*}
  \sum_{j}\parab{\Phi_i\Phi_j\Phi_i^{-1}\Phi_j^{-1} - \Phi_j^{-1}\Phi_i\Phi_j\Phi_i^{-1} + \Phi_i\Phi_j^{-1}\Phi_i^{-1}\Phi_j-\Phi_j\Phi_i\Phi_j^{-1}\Phi_i^{-1}} = 0
\end{align*}
for all $i$.
\begin{conjecture}
  If $N\to+\infty$, limits of critical values of $S_q$ are either
  $+\infty$, or positive $\integers$-linear combinations of critical values of $S_{cl}$. Moreover, if the limit is $<\infty$ then  $\norm{\Phi_i\Phi_j-\Phi_j\Phi_i}\leq\frac{\operatorname{const}}{N}$.
\end{conjecture}

One can also propose a rough criterium for a sequence of quantum maps to torus depending on dimension $N\to +\infty$ to ``approximate'' a given  oriented surface $M\subset (S^1)^n$ endowed with a symplectic 2-form $\omega$ :
\begin{align}
\frac{1}{\int_M \omega}\int_M e^{i \sum_j  k_j \theta_j}=\frac{1}{N} \text{Tr} \,\,\Phi_1^{k_1}\cdot \dots\cdot \Phi_n^{k_n}+O(1/N),\quad \forall (k_1,\dots,k_n)\in \integers^n
\end{align}

\noindent 
\underline{Generalizations}: Fix symmetric positive $n\times n$ matrix
$g=(g_{ij})_{1\leq i,j\leq n}$, $g=g^\dagger$ and $g>0$. $g$ gives a
flat (constant) metric $\sum g_{ij}d\phi_id\phi_j$ on $(S^1)^n$. Weierstrass parametrization $\rightsquigarrow$ explicit description of critical values of $S_{cl}$ in terms of complex curves. Critical values of $S_q$ should approximate these for $N\to +\infty$.

\subsection{IV -- Calibrated geometry}

For all $n\geq 0$ the Yang-Mills algebra
is given by
\begin{align}\label{eq:YM.algebra.def}
  YM_n = \complex\angles{X_1,\ldots,X_n}\slash\parab{\sum_{j}[X_j,[X_j,X_i]] = 0\quad\forall i}
\end{align}
which is a $\ast$-algebra upon setting
$X_i^\dagger=X_i$. \underline{Reason}: if
$X_i\in\Mat(N\times N,\complex)$ are self-adjoint the connection on the
trivial $U(N)$-bundle$\slash\reals^n$ with \emph{constant
  coefficients}
\begin{align*}
  \nabla = d+i\sum_{j}X_jdx_j
\end{align*}
(where $x_1,\ldots,x_n$ are coordinates in $\reals^n$) satisfies
Yang-Mills equation $\leftrightarrow$ \eqref{eq:YM.algebra.def}.
\begin{align*}
  \text{YM action}\sim\sum_{i<j}\Tr\underbrace{[X_i,X_j]^2}_{\text{curvature form}}
\end{align*}
critical points $\Leftrightarrow \eqref{eq:YM.algebra.def}$. If $n=2m$ Hermitian YM equation; some equations of simpler form implying YM.
\begin{align*}
  &Z_k = X_k+iX_{m+k}\\
  &Z_k^\dagger = X_k-iX_{m+k}
\end{align*}
for $i=1,\ldots,m$.
\begin{equation}\label{eq:HYM.algebra.def}
  \begin{split}
    \text{HYM algebra} =
    \complex&\angles{Z_1,\ldots,Z_m,Z_1^\dagger,\ldots,Z_m^\dagger}\slash\\
    &\parab{\sum_{k}[Z_k^\dagger,Z_k]=0\text{ and }[Z_j,Z_k]=0=[Z_j^\dagger,Z_k^\dagger]}
  \end{split}
\end{equation}
Indeed: \eqref{eq:HYM.algebra.def} implies that
\begin{align*}
  \sum_{j}[Z_j,[Z_j^\dagger,Z_i]]+\sum_{j}[Z_j^\dagger,[Z_j,Z_i]]=0.
\end{align*}
\underline{Meaning of HYM}: On $\reals^{2m}=\complex^m$,
$U(N)$-connection in trivial bundle with constant coefficients.
\begin{align*}
  &[Z_i,Z_j]=0=[Z_i^\dagger,Z_j^\dagger]\qquad \forall i<j\quad\Rightarrow
    \text{ Holomorphic bundle}\\
  &\sum_{i}[Z_i^\dagger,Z_i] = 0\quad\Leftrightarrow\quad
    \text{curvature form}\cdot(\omega^{1,1})^{m-1}=0
\end{align*}
Recall that  by Donaldson-Uhlenbeck-Yau theorem, any stable holomorphic bundle with first Chern class 0 on compact K\"ahler manifold admits a unique (up to scalar) Hermitian metric whose canonically associated connection satisfies Hermitian YM equation.

\underline{Generalization:}
\begin{align*}
  &\sum_{i}[Z_i^\dagger,Z_i] = c\cdot 1\in\reals\quad\text{(constant)}\\
  &[Z_i,Z_j] = c_{ij}\in\complex\qquad c_{ij}=-c_{ji} \quad\text{(constant)}\\
  &[Z_i^\dagger,Z_j^\dagger] =\bar{c}_{ij}
\end{align*}
gives a Lie algebra with a central extension.

\underline{Classical limit of HYM}. It looks very much reasonable that
the classical limit of representations of HYM algebra (possibly with
central extension as above given by $c\ne 0$, but still with
$c_{ij}=0$) correspond to a special class of minimal surfaces in
$\reals^{2m}=\complex^m$ which are holomorphic curves.

Recall that for a surface in any K\"ahler manifold the property to be
a holomorphic curve (which is a first order constraint) implies that
the surface is minimal (which is a second order constraint). This is
the simplest case of so called \emph{calibrated geometry}.  Similarly,
relations in HYM algebra are identities between \emph{single}
commutators, whereas in YM algebra we have \emph{double} commutators.

\underline{Case:} $m=2$, $n=4$, $c=0$ and $c_{ij}=0$. Nice Lie algebra
SYM (self-dual Yang Mills in $\dim=4$ constant solutions).
\begin{align*}
  &[X_4,X_1]=[X_2,X_3]\\
  &[X_4,X_2] = [X_3,X_1]\\
  &[X_4,X_3] = [X_1,X_2]\equivalent\\
  &[X_i,X_j] = \frac{1}{2}\sum_{k,l}\varepsilon_{ijkl}[X_k,X_l]\qquad SO(4)-\text{covariant}
\end{align*}
In \cite{ibs:tbranes} these are called Banks-Seiberg-Shenker
equations (\cite{bss:branes.matrices}), and unlike in the general case
$m\ge 3$, they admit an enhanced symmetry $SO(4,\reals)\supset U(2)$,

In general, one can ask the following question: for a given
finitely-generated module $M$ over $\complex[Z_1,\dots,Z_m]$ (which is
the same as an algebraic coherent sheaf on $\complex^m$) can one
construct a pre-Hilbert space $\mathcal M_\infty$ containing $M$ and
such that $M$ is dense in $\mathcal M_\infty$ for such that operators
of multiplication by $Z_i$ extend to $M_\infty$ and admit hermitian
conjugate, and satisfy the relation of HY algebra (possibly centrally
extended by $c$).  The question also makes sense when algebra
$\complex[Z_1,\dots,Z_m]$ is replaced by its quantum deformation
$[Z_i,Z_j]=c_{ij}$.

For example, if $M=\complex[Z_1,\dots,Z_m]$ is the free module of rank
1 (i.e. a trivial rank one bundle in terms of sheaves), the
pre-Hilbert space $M_\infty$ is the space of entire functions
satisfying certain growth condition:
\[ M_\infty=\{\sum_{k_1,\dots,k_m\ge 0} c_{k_1,\dots,k_m}Z_1^{k_1}\dots Z_m^{k_m}|   \,\,|c_{k_1,\dots,k_m}|=O\big(\prod_j \frac{1}{\sqrt{k_j !}}\cdot (1+\sum_j k_j)^r\big) \,\,\forall r>0\}\]
and with the pre-Hilbert norm given by
\[|f|^2:=\int_{\vec{Z}\in \complex^m} |f(\vec{Z})|^2 \exp(-\sum_j |Z_j|^2) \prod d\Re(Z_j) d\Im(Z_j),\quad \vec{Z}=(Z_1,\dots,Z_m)\in \complex^m\]

Nekrasov considered solution of quantum HYM which correspond to
bundles or torsion-free sheaves in $\complex^m$ (and also for modules
close to free for the quantized $\complex^m$).
By our philosophy, his solutions do not correspond to minimal
surfaces (or complex curves), as supports in the classical limit are
full $2m$-dimensional.

From our perspective, the (noncompact) minimal surfaces in
$\reals^{2m}=\complex^m$ which are complex algebraic curves, should
correspond to solution of HYM for coherent sheaves supported on
curves, e.g. quotient modules of $\complex[Z_1,\dots,Z_m]$ by the
ideal generated by defining equations of a curve.  The problem of
constructing quantum analogs of algebraic curves in $\complex^m$ was
considered in \cite{ct;holomorphic.curves}, without stressing the
relation to minimal surfaces.

Now we discuss possibilities to have ``calibrated" quantum minimal
surfaces in the compact case.

\underline{Unitary version, case $m=2$}: $\Phi_1,\Phi_2,\Phi_3,\Phi_4\in U(N)$
\begin{equation}\label{eq:unitary.4Phi}
  \begin{split}
    &\Phi_4\Phi_1\Phi_4^{-1}\Phi_1^{-1}=\Phi_2\Phi_3\Phi_2^{-1}\Phi_3^{-1}\\
    &\Phi_4\Phi_2\Phi_4^{-1}\Phi_2^{-1}=\Phi_3\Phi_1\Phi_3^{-1}\Phi_1^{-1}\\
    &\Phi_4\Phi_3\Phi_4^{-1}\Phi_3^{-1}=\Phi_1\Phi_2\Phi_1^{-1}\Phi_2^{-1}
  \end{split}
\end{equation}
gives a $A_4$-covariant group ($A_4$ = even permutations in
$S_4$). Looks like $\pi_1(\text{3-dim manifold})$, \#relations = 3 =
\#generators$-1$.

We expect critical points of
\begin{align*}
  N\cdot\sum_{1\leq i<j\leq 4}\Tr(2-\Phi_i\Phi_j\Phi_i^{-1}\Phi_{j}^{-1}-
  \Phi_j\Phi_i\Phi_j^{-1}\Phi_i^{-1})
\end{align*}
approximate (as $N\to +\infty$) minimal surfaces in flat torus
$(S^1)^4$ (second order equation).
\begin{conjecture}
  Representations of group \eqref{eq:unitary.4Phi}
  \begin{align*}
    \Phi_i\Phi_j\Phi_i^{-1}\Phi_j^{-1}=\Phi_k\Phi_l\Phi_k^{-1}\Phi_l^{-1}
  \end{align*}
  ($(ijkl)$ and even permutation) approximate as $N\to+\infty$ complex
  curves in the abelian variety
  $\paraa{\complex\slash\integers+i\integers}^2$ (first order equation).
\end{conjecture}

\noindent

\subsection{V -- Quantum curves in abelian varieties:}

Here is another possibility, valid in arbitrary complex dimension $m$
and for any translationally invariant complex structure and a K\"ahler
metric on $(S^1)^{2m}$.

Let $\Gamma\simeq \integers^{2m}\subset \complex^m$ be a lattice, such
that the quotient torus $A=\complex^m/\Gamma$ (here we use the
standard K\"ahler metric in the coordinate space $\complex^m$)
contains a compact complex algebraic curve $C$, possibly singular.
Notice that if $C$ is not degenerate in the sense that none of its
irreducible components is not contained in a proper complex subtorus,
then $A$ is in fact algebraic and is an abelian variety.

Our goal is to construct an infinite sequence of finite-dimensional
Hilbert spaces with dimension $\to +\infty$, endowed with
``almost-commuting'' $2m$ unitary operators which approximates in some
sense $C$, or, more generally, a coherent sheaf on $A$ supported on
$C$.

The idea is the following. Let us pass to the universal cover
$\complex^m$ of $A$, then the pre-image of $C$ will be a
\emph{non-algebraic} curve $\widetilde{C}\subset \complex^m$,
invariant under shifts by $\Gamma$.  Optimistically extending our
previous considerations of quantization of curves to the non-algebraic
case, we are lead to the following question. Construct an
infinite-dimensional representation of HYM algebra
\begin{align*}
  &[Z_i,Z_j] = 0\,\forall i<j\\
 &\sum_{i}[Z_i^\dagger,Z_i] = c\cdot 1\in\reals
\end{align*}
covariant with respect to the action of $\Gamma$, i.e. assuming that
we are also given a collection of commuting unitary operators
$U_1,\dots, U_{2m}$ corresponding to a basis of $\Gamma$, satisfying
commutator relations
\[U_k^{-1} Z_j U_k=Z_j+a_{jk}\quad \forall j=1,\dots,m,\,k=1,\dots,
  2m\] where $a_{jk}\in \complex$ are coefficients of generators of
$\Gamma$ considered as vectors in $\complex^m$.

We assume that the subspace $\mathcal H_{u_1,\dots,u_{2m}}$ of the
ambient Hilbert space $\mathcal H$ corresponding to eigenvalues
$u_1,\dots,u_{2m}\in U(1)\subset \complex^\times$ of operators
$U_1,\dots, U_{2m}$, is \emph{finite-dimensional}. This will be our
finite-dimensional Hilbert space ``approximating'' $C\subset A$. The
almost-commuting unitary operators acting on this space will be of the
form
\begin{align}\label{eq:unitary.almost.commute}
  \exp\left(\sum_j (b_j Z_j- \bar{b_j}Z_j^\dagger)\right)\end{align}
where $\vec{b}=(b_1,\dots,b_m)\subset (\complex^m)^*$ belongs to the dual lattice determined by the constraint
\[U_k^{-1}\cdot \sum_j(b_j Z_j -\bar{b_j}Z_j^\dagger) \cdot U_k= \sum_j(b_j Z_j -\bar{b_j}Z_j^\dagger) + \text{ element of } 2\pi i \integers\quad \forall k=1,\dots,2m\]

There is a natural proposal to define $\mathcal H$ based on
Fourier-Mukai duality. Namely, Hilbert space $\mathcal H$ (or, more
precisely, certain dense pre-Hilbert subspace
$\mathcal H_\infty\subset \mathcal H$ corresponding to ``Schwarz
functions'') is a finitely generated projective module over the
algebra of functions $C^\infty(T)$ on the torus of unitary characters
$\vec{u}=(u_1,\dots,u_{2m})$ of lattice $\Gamma$. In other words,
$\mathcal H_\infty$ is the space of section of a complex vector bundle
$\mathcal E$ on $T$, the pre-Hilbert space structure is given by a
hermitian structure on $\mathcal E$ (and integration with respect to
the Haar measure on $T$). Operators $Z_i$ in this presentation are
first order differential operators, corresponding to the covariant
derivatives in (anti)-holomorphic directions with respect to certain
unitary connection on $\mathcal E$.  The condition that we have a
representation of HYM algebra, translates to the condition that we
have a solution of the usual Hermitian Yang-Mills equations from
differential geometry (this is a classical idea of $T$-duality for
solutions of noncommutative YM equations on tori, see
e.g. \cite{cds:compactification}).

Notice that almost-commuting unitary operators
$\exp(\sum_j (b_j Z_j- \bar{b_j}Z_j^\dagger)$ introduced formally in
\eqref{eq:unitary.almost.commute}, have geometric meaning as holonomy
operators for the connection along geodesic loops in $T$.

By Donaldson-Uhlenbeck-Yau theorem, we see that $\mathcal E$ is a
(semi-)stable holomorphic vector bundle on $T$, which as a complex
manifold is the same as the \emph{dual} abelian variety
$A^\vee:=Pic_0(A)$. Here is our proposal: take $\mathcal E$ to be the
Fourier-Mukai dual to the coherent sheaf $\mathcal F\in Coh(A)$
{supported} on $C\subset A$. In order to introduce a small parameter
(the inverse rank of $\mathcal E$) we can consider $\mathcal F$ of the
form
\[\mathcal F=\mathcal F_n:= \mathcal F_0\otimes \mathcal L_{|C}^{\otimes n},\quad n\to +\infty\]
where $\mathcal F_0$ is a coherent sheaf supported on $C$, and
$\mathcal L$ is an ample line bundle on $A$.

\subsection{VI -- Quantum degree}

Calibrated submanifolds (like complex curves in K\"ahler manifolds)
have the characteristic property that their area depends only on their
homology class, and for them the calibrating lower bound on the area
in given homology class (BPS bound in physics), is saturated. Hence,
we are lead to the question what is the ``homology class" of a
``quantum surface".  Here are two simple examples when one can define
such a class, which is the quantum degree of a map from the quantum
surface to the target surface.

We did not study yet the relation of quantum degree with quantum
calibrated geometry of minimal surfaces discussed above.

\begin{enumerate}
\item $d=1$: Smooth map $\varphi:M\to S^1\times S^1=U(1)\times U(1)$
  is given by $\varphi_1,\varphi_2\in C^{\infty}(M,\complex)$ with
  $|\varphi_1|^2=|\varphi_2|^2=1$. The degree
  $\deg(\varphi)\in\integers$ is given by
  \begin{align*}
    \deg(\varphi) = \frac{1}{(2\pi i)^2}\int_M\{\log\varphi_1,\log\varphi_2\}.
  \end{align*}
  \begin{proposition}
    Let $\Phi_1,\Phi_2\in U(N)$ ($N\to\infty$) be unitary matrices
    depending on $N$ and assume that there exists $c>0$ such that
   \begin{align}\norm{\Phi_1\Phi_2\Phi_1^{-1}\Phi_2^{-1}-\mid_{N\times N}}\leq\frac{c}{N}\label{eq:conj.phi.phi}
     \end{align}
    Then 
    \begin{align}\label{eq:conj.phi.phi.k}
      \Tr(\Phi_1\Phi_2\Phi_1^{-1}\Phi_2^{-1}-\mid_{N\times N})
      =2\pi i k + o(1)
    \end{align}
    where $k\in\integers$ and $|k|<c/2\pi\cdot (1+o(1))$.
    Note that
    \begin{align*}
      &\text{\eqref{eq:conj.phi.phi}}\equivalent
      \norm{\Phi_1\Phi_2-\Phi_2\Phi_1}\leq\frac{c}{N}\\
      &\text{\eqref{eq:conj.phi.phi.k}}\equivalent
        e^{\Tr(\Phi_1\Phi_2\Phi_1^{-1}\Phi_2^{-1}-\mid_{N\times N})}=1+o(1).
    \end{align*}
  \end{proposition}

  \noindent
 
``Explanation'': Assume $\Phi_i=\hat{\varphi}_i$ for $i=1,2$.
  \begin{align*}
    \Tr\paraa{\Phi_1\Phi_2\Phi_1^{-1}\Phi_2^{-1}-\mid_{N\times N}}\sim
    \frac{1}{2\pi\hbar}\int_Mi\hbar\{\log\varphi_1,\log\varphi_2\}
    =-2\pi i\deg(\varphi)
  \end{align*}

 \begin{proof}
Let us denote eigenvalues of unitary operator $\Psi:=\Phi_1\Phi_2\Phi_1^{-1}\Phi_2^{-1}$ by 
\[ u_j=\exp(i\theta_j)\subset \complex^\times,\text{ where }|\theta_j|\le \pi, \quad 1\le j \le N, \]
 Assumption \eqref{eq:conj.phi.phi}
means that $|u_j-1|\le c/N$ for all $j$, hence $|\theta_j|<c/N\cdot(1+o(1))$. Obviously, we have $\det \Psi=1$,
therefore $\sum_j \theta_j=2\pi k $ for some $k\in  \integers$. Inequality  $|\theta_j|<c/N\cdot(1+o(1))$ implies that  $|k|<c/2\pi(1+o(1))$. Finally, expression on the l.h.s. of \eqref{eq:conj.phi.phi.k} is equal to
\[\sum_j (\exp(i\theta_j)-1)=i\sum_j \theta_j+\sum_j O(\theta_j^2)=2\pi i k+O(1/N)\]

\end{proof}

\item Analogously, a smooth map $\varphi:M\to$ 2-dim sphere
  $S^2\subset\reals^3$ is given by $x_1,x_2,x_3\in C^\infty(M,\reals)$
  with $x_1^2+x_2^2+x_3^2=1$. The degree $\deg(\varphi)\in\integers$
  is given by
  \begin{align*}
    \deg(\varphi) = \frac{1}{4\pi}
    \int_M\parab{x_1\{x_2,x_3\}+x_2\{x_3,x_1\}+x_3\{x_1,x_2\}}=\frac{3}{4\pi}
    \int_M x_1\{x_2,x_3\}
  \end{align*}
  \begin{conjecture}
    Let $X_1,X_2,X_3\in\Mat(N\times N,\complex)$ be self-adjoint
    operators (depending on $N$) such that
    $X_1^2+X_2^2+X_3^2=\mid_{N\times N}$, and assume that there exist
    $c$ such that
    \begin{align*}
      \norm{X_iX_j-X_jX_i}\leq\frac{c}{N}
    \end{align*}
    for $i,j\in\{1,2,3\}$. Then
    \begin{align*}
      \Tr\paraa{X_1[X_2,X_3]}=\frac{2i}{3}k+o(1)
    \end{align*}
    where $k\in\integers$ and $|k|\leq 3/2c\cdot (1+O(1/N))$.
  \end{conjecture}
\end{enumerate}

\section{Quantized complex curves and integrable systems}\label{sec:quantized.complex.curves}

\noindent
As shown already in the 19th century (see
e.g. \cite{k:kruemmung,e:minimal.four.space}),
\begin{align}
  x_3+ix_4 = f(x_1+ix_2)
\end{align}
defines a minimal surface in $\reals^4$ for arbitrary analytic
$f$. Similarly, static membranes (solutions of
\eqref{eq:X.double.com}), in particular a complex parabola, were
considered in \cite{ct;holomorphic.curves} (curiously without explicitly
mentioning ``minimal surfaces'') for which
\begin{align}
  Z_2:=X_3+iX_4=f(X_1+iX_2)=f(Z_1)
\end{align}
and (cp. Section \ref{sec:remarks})
\begin{align}\label{eq:ZZepsilon}
  [Z_1^\dagger,Z_1] + [Z_2^\dagger,Z_2]=\epsilon \mid.
\end{align}
Before taking up that parabola example (deriving many new properties,
and noting that it constitutes a discrete integrable system,
cp. \cite{h:diophantine}) let us first mention
(cp. Section~\ref{sec:hyperbola}) the simpler (though previously
unnoticed), most beautiful, example: the complex hyperbola
\begin{align}
  Z_1\cdot Z_2=c\mid,
\end{align}
yielding the recursion relations
\begin{align}
  r_n-r_{n+1}+\frac{|c|^2}{r_{n+1}}-\frac{|c|^2}{r_n}=\epsilon
\end{align}
with $r_n=|w_n|^2$, $Z_1\ket{n}=w_n\ket{n-1}$ for
$n\in\integers$. Solving the quadratic equation
\begin{align}
  r_n-\frac{|c|^2}{r_n}=-\epsilon n+\delta
\end{align}
gives (compare \eqref{eq:rnwngamma})
\begin{align}
  r_n = \frac{-\epsilon n+\delta}{2}\pm
  \sqrt{\parac{\frac{-\epsilon n+\delta}{2}}^2+|c|^2}
\end{align}
as an exact solution of \eqref{eq:ZZepsilon},
resp. \eqref{eq:X.double.com}. The classical limit of this Quantum
Curve is the minimal surface described in Section \ref{sec:hyperbola},
where $c$ was taken for simplicity to be real; for complex
$c=\alpha+i\beta$ a parametrization of the real 4-dimensional
embedding is
\begin{align}
  \xv(t)=\paraa{t\cos(u),t\sin(u),
  \tfrac{1}{t}(\alpha\cos(u)+\beta\sin(u)),
  \tfrac{1}{t}(-\alpha\sin(u)+\beta\cos(u))}_{t>0}.
\end{align}
Let us now apply our general method, explained in the introduction, to
the complex parabola (considered in \cite{ct;holomorphic.curves})
$Z_2=Z_1^2$, respectively:
\begin{align}
  &\xv(r,u) = \paraa{r\cos(u),r\sin(u),r^2\cos(2u),r^2\sin(2u)},
\end{align}
$\dot{\xv} = \paraa{\cos(u),\sin(u),2r\cos(2u),2r\sin(2u)}$,
$\xv' = \paraa{-r\sin(u),r\cos(u),-2r^2\sin(2u),2r^2\cos(2u)}$ implies
\begin{align*}
  (g_{ab}) =
  \begin{pmatrix}
    1+4r^2 & 0\\
    0 & r^2+4r^4
  \end{pmatrix}
        =(1+4r^2)
        \begin{pmatrix}
          1 & 0 \\ 0 & r^2
        \end{pmatrix}
\end{align*}
Reparametrizing $\xv(r,u)=\tilde{\xv}(\tilde{r},u)$ according to
(cp. \eqref{eq:J})
\begin{equation}
  \begin{split}
    &\frac{d\tilde{r}}{dr} = \sqrt{g}(r) =r(1+4r^2)\implies
    \tilde{r} = \frac{1}{2}r^2+r^4-c\\
    &r^2 = -\frac{1}{4}+\sqrt{\frac{1}{16}+\tilde{r}+c}=\paraa{r(\tilde{r})}^2
  \end{split}
\end{equation}
shows that, with
$e^{i\Uh}\ket{n}=\ket{n+1}\hat{=}e^{i\varphi}\ket{e^{in\varphi}}$, $\hat{r}=r(\hat{\tilde{r}}=-i\hbar\frac{\d}{\d\varphi})$
\begin{align}
  W\ket{n}=w_n\ket{n+1} \qquad n=0,1,2,\ldots,
\end{align}
the condition
\begin{align}\label{eq:WWd.sq.mid}
  [\Wd,W] + [{\Wd}^2,W^2]\sim\mid,
\end{align}
resp. (cp. \cite{ct;holomorphic.curves}, eq. (25)), with $|w_n|^2=v_n$,
\begin{align}\label{eq:alphan.recursion}
  v_n-v_{n-1}+v_{n+1}v_n-v_{n-1}v_{n-2}=\epsilon
\end{align}
should ``approximately'' (s.b., including the integration constant $c$)
be solved by
\begin{align}\label{eq:alphan.formula}
  v_n = -\frac{1}{4}+\sqrt{\frac{1}{16}+n\hbar+c}=r_n^2=r^2(-n\hbar).
\end{align}
\eqref{eq:alphan.recursion}, resp.
\begin{align}\label{eq:alphan.short.rec}
  v_n\paraa{v_{n+1}+v_{n-1}+1}=2\hbar n+\delta
\end{align}
is to the first orders, indeed solved by \eqref{eq:alphan.formula},
both for $n\hbar\to 0$ (i.e. $n<n_0$, $\hbar\to 0$) and
$n\hbar\to\infty$ (resp. $\hbar$ fixed, $n\to\infty$) which is easily
verified by inserting the approximations of
\begin{align}
  v_n = \alpha + \sqrt{\beta^2+n\hbar} = \alpha+\beta\sqrt{1+\frac{n\hbar}{\beta^2}}
\end{align}
i.e.
\begin{align}
  v_n = \alpha+\beta
  \parac{1+\frac{n\hbar}{2\beta^2}-\frac{n^2\hbar^2}{8\beta^4}+\cdots}
  =\gamma+\frac{n\hbar}{2\beta}-\frac{n^2\hbar^2}{8\beta^3}
\end{align}
resp.
\begin{align}
  v_n = \alpha + \sqrt{n\hbar}\sqrt{1+\frac{\beta^2}{n\hbar}}
  =\sqrt{n\hbar}+\alpha+\frac{\beta^2}{2\sqrt{n\hbar}}
\end{align}
yielding (cp. \eqref{eq:alphan.short.rec}, $\epsilon=2\hbar$)
\begin{align}
  \delta+2\hbar n = \gamma\parac{2\gamma+1-\frac{\hbar^2}{4\beta^3}}
  +\frac{n\hbar}{2\beta}\parac{4\gamma+1-\frac{\hbar^2}{4\beta^3}}
  -\frac{n^2\hbar^2}{8\beta^3}\parac{4\alpha+1-\frac{\hbar^2}{4\beta^3}}+O(n^3)
\end{align}
i.e. (ignoring the $\hbar^2/4\beta^3$ correction) $4\gamma+1=4\beta$
($\Rightarrow$ $\alpha=-1/4$, which makes the $O(n^2\hbar^2)$ term
vanish) and
\begin{align}
  &\delta = \gamma(2\gamma+1)=\paraa{\beta-\tfrac{1}{4}}\paraa{2\beta+\tfrac{1}{2}}
  =2\paraa{\beta^2-\tfrac{1}{16}}=2c\\
  &\alpha=-\frac{1}{4}\notag
\end{align}
resp.
\begin{align}
  \delta+2\hbar n = 2\hbar n+\sqrt{n}\paraa{2\alpha+1+2\alpha}+
  \parab{\alpha(2\alpha+1)+2\beta^2}+\cdots
\end{align}
i.e., again,
\begin{align}
  \alpha = -\frac{1}{4},\quad
  &-\frac{1}{8}+2\beta^2=\delta\\
  &2\paraa{\beta^2-\tfrac{1}{16}}=2c.\notag
\end{align}
Finally note that while \eqref{eq:alphan.recursion}
($\epsilon=2\hbar$) for $n=0$ ($v_{-1}=0=v_{-2}$) gives
$v_0(1+v_1)=2\hbar$, \eqref{eq:alphan.short.rec} (for $n=0$)
implies $v_0(1+v_1)=\delta$, hence
\begin{align}
  \delta=2\hbar=\epsilon, \quad c=\hbar,
\end{align}
in which case
\begin{align}
  v_n = -\frac{1}{4}+\sqrt{\frac{1}{16}+n\hbar+\hbar}
\end{align}
automatically vanishes at $n=-1$ (cp. the Enneper case), and gives
\begin{align}
  v_0 = -\frac{1}{4}+\sqrt{\frac{1}{16}+\hbar},\label{eq:alpha.zero}
\end{align}
which is $1/2$ for $\hbar=1/2$ (the small discrepancy with the in
\cite{ct;holomorphic.curves} numerically calculated value
$\rho_0^2=v_0(\hbar=\tfrac{1}{2})\gtrapprox
\tfrac{9}{16}\gtrapprox\tfrac{1}{2}$
could partly be due to the non-negligible value of
$\tfrac{\hbar^2}{4\beta^3}=\frac{1}{4^3(\frac{3}{4})^3}=\frac{1}{27}$
in this case).

Note that \eqref{eq:alpha.zero} gives
\begin{align*}
  v_1 = 2v_0,\quad v_2 = v_0,\quad v_3 = 4v_0+2, v_4<0, 
\end{align*}
so that it does not correspond to an allowed initial value (as
$v_n=|w_n|^2$ must necessarily be non-negative), although as we will
see, $\epsilon-2\epsilon^2+8\epsilon^3+O(\epsilon^4)$ is rather
close.

Let us calculate the first few $v_n$'s, from
\begin{align}\label{eq:vn.recursion}
  v_n\paraa{v_{n+1}+v_{n-1}+1} = \epsilon(n+1),
\end{align}
for arbitrary $v_0=:x=:u_0$ (always taking $v_{-1}=0$):
\begin{align}
  &v_1 = \frac{\epsilon-x}{x}=:\frac{u_1}{u_0}\implies
  x\in I_1=(0,\epsilon)=:(c_0,c_1)\notag\\
  &v_2 = \frac{x^2+x(1+\epsilon)-\epsilon}{\epsilon-x}
    =\frac{(x-v_+^{(2)})(x-v_-^{(2)})}{u_1}=:\frac{u_2}{u_1}\label{eq:v.two}\\
  &v_{\pm}^{(2)}=\frac{1+\epsilon}{2}\parac{\pm\sqrt{1+\frac{4\epsilon}{(1+\epsilon)^2}}-1}\notag\\
  &\implies
    x\in I_2=(c_2,c_1),\quad c_2=v_+^{(2)}=\epsilon-2\epsilon^2+6\epsilon^3+O(\epsilon^4)\notag
\end{align}
\begin{align}
  &v_3 = \frac{-\epsilon}{x}\frac{(4x^2+x(1-2\epsilon)-\epsilon)}{x^2+x(1+\epsilon)-\epsilon}
  =\frac{-4\epsilon}{x}\frac{(x-v_+^{(3)})(x-v_-^{(3)})}{(x-v_+^{(2)})(x-v_-^{(2)})}=:\frac{u_3}{u_0u_2}\label{eq:v.three}\\
  &v_{\pm}^{(3)} = \frac{1-2\epsilon}{8}\parac{\pm\sqrt{1+\frac{16\epsilon}{(1-2\epsilon)^2}}-1}\notag\\
  &\implies x\in I_3=(c_2,c_3),\quad c_3=v_+^{(3)}=\epsilon-2\epsilon^2+12\epsilon^3+O(\epsilon^4)\notag
\end{align}
\begin{align}
  &v_4 = \frac{x}{\epsilon-x}\frac{\paraa{(3-2\epsilon)x^2-2x\epsilon(\epsilon+4)+5\epsilon^2}(-\epsilon)}
  {\paraa{4x^2+x(1-2\epsilon)-\epsilon}(-\epsilon)}=:\frac{u_0u_4}{u_1u_3}\label{eq:v.four}\\
  &u_4(x) = -\epsilon(3-2\epsilon)(x-v_+^{(4)})(x-v_-^{(4)})\notag\\
  &v_{\pm}^{(4)} = \frac{\epsilon(4+\epsilon)}{3-2\epsilon}\parac{\sqrt{1-\frac{5(3-2\epsilon)}{(\epsilon+4)^2}}+1}\notag\\
  &\implies x\in I_4=(c_4,c_3),\quad
    c_4 = v_-^{(4)}=\epsilon-2\epsilon^2+\parab{\tfrac{19}{2}-\tfrac{1}{144}}\epsilon^3+O(\epsilon^4)\notag\\
  &v_n(x) = \frac{u_n}{u_{n-3}}\cdot\frac{u_{n-4}}{u_{n-1}}\label{eq:vnx}
\end{align}
As indicated above, let $I_n\subset\reals$ denote an open interval on
which $v_n(x)$ is positive. We will now construct intervals
$I_n\subset I_{n-1}$ such that
$\bigcap_{n\in\naturals}I_n\neq \emptyset$, proving that there exist
at least one initial condition $x_0$ such that $v_n(x_0)>0$ for all
$n\geq 0$. To this end, we assume that $c_{2n-2}<c_{2n}<c_{2n-1}$ and
\begin{alignat}{2}
  &\quad I_{2n-1}=(c_{2n-2},c_{2n-1}) &\qquad
  &\quad I_{2n}=(c_{2n},c_{2n-1})\notag\\
  &\lim_{x\to c_{2n-2}}v_{2n-1}(x)=+\infty &
  &\lim_{x\to c_{2n-1}}v_{2n-1}(x) = 0\label{eq:vn.In.induction.assumption}\\
  &\lim_{x\to c_{2n}}v_{2n}(x)=0&
  &\lim_{x\to c_{2n-1}}v_{2n}(x)=+\infty\notag
\end{alignat}
where the limit points are approached from inside the intervals. It
follows from \eqref{eq:v.two}--~\eqref{eq:v.four} that these
conditions are met for $n=1,2$. From the recursion relation one
finds
\begin{align*}
  v_{2n+1} = \frac{(2n+1)\epsilon}{v_{2n}}-v_{2n-1}-1
\end{align*}
implying that
\begin{align*}
  \lim_{x\to c_{2n}}v_{2n+1}(x) = +\infty\qquad
  \lim_{x\to c_{2n-1}}v_{2n+1}(x) = -1.
\end{align*}
Hence, there exists $c_{2n+1}\in(c_{2n},c_{2n-1})$ such that
$v_{2n+1}(c_{2n+1})=0$ and $v_{2n+1}(x)>0$ for
$x\in I_{2n+1}=(c_{2n},c_{2n+1})$. Analogously,
\begin{align*}
  v_{2n+2} = \frac{(2n+2)\epsilon}{v_{2n+1}}-v_{2n}-1
\end{align*}
implies that
\begin{align*}
  \lim_{x\to c_{2n+1}}v_{2n+2}(x)=+\infty\qquad
  \lim_{x\to c_{2n}}v_{2n+1}(x)=-1
\end{align*}
and that there exists $c_{2n+2}\in (c_{2n},c_{2n+1})$ such that
$v_{2n+2}(c_{2n+2})=0$ and $v_{2n+2}(x)>0$ for
$x\in I_{2n+2}=(c_{2n+2},c_{2n+1})$. By induction, we conclude that
\eqref{eq:vn.In.induction.assumption} holds true for all $n$ and note that
\begin{align*}
  I_{n+1}\subset I_n\quad\text{ and }\quad\overline{I_{n+2}}\subset I_{n}
\end{align*}
implying that $\bigcap_{n\in\naturals}I_n\neq\emptyset$.
The uniqueness of the initial condition (conjectured in
\cite{ct;holomorphic.curves}, based on numerical findings)
\begin{align*}
  \bigcap_{n\in\integers} I_n = \{\hat{v}\},
\end{align*}
probably follows similarly (cp. \cite{clva:discrete.painleve}).
Writing \eqref{eq:alphan.recursion}, resp \eqref{eq:vn.recursion}, in
terms of $u_n$ resp.
\begin{align}
 r_n:=\tfrac{u_n}{u_{n-3}}\quad\text{(resp. }\alpha_n=\tfrac{r_n}{r_{n-1}})
\end{align}
gives
\begin{align}
  &\epsilon r_{n-1}r_{n-2}r_{n-3}+r_{n-1}^2(r_{n-2}+r_{n-3})=r_{n-2}r_{n-3}(r_n+r_{n+1})\label{eq:rrr.rec}\\
  &u_{n+1}u_{n}u_{n-1}u_{n-2}=u_{n-2}^2\paraa{u_{n+3}u_{n-1}+u_{n+2}u_n}
    -u_{n+1}^2\paraa{u_nu_{n-4}+u_{n-1}u_{n-3}}\notag\\
  &\epsilon(n+1)r_{n-1}r_{n-2} = r_nr_{n-1}+r_{n-2}(r_n+r_{n+1})\label{eq:rrr.sec.rec}\\
  &\epsilon(n+3)u_{n+2}u_nu_{n-1}=u_{n+3}u_{n-1}u_{n-2}
    +u_{n+2}\paraa{u_{n+1}u_{n-3}+u_nu_{n-2}}\notag,
\end{align}
i.e homogenizing the recursion relations, while increasing the number
of terms needed to calculate the next. Noting
\begin{align}
  \frac{u_n}{u_{n-3}}=r_n=\bra{0}{\Wd}^{n+1}W^{n+1}\ket{0}
\end{align}
it is tempting to try to prove the polynomiality (in $x=\bra{0}\Wd W\ket{0}$)
\begin{align}
  \bra{0}{\Wd}^nW^n\ket{0}\bra{0}{\Wd}^{n-3}W^{n-3}\ket{0}\cdots
\end{align}
compactly, via \eqref{eq:WWd.sq.mid}. Having found the substantial
cancellations leading to the extremely slow (``integrable'') growth of
the degrees of numerator and denominator of the rational function
$v_n(x)$ (cp. \eqref{eq:vnx}), the ``focusing'' of zeroes of the $u_n$
(note that also those building blocks are increasing for even $N$,
decreasing for odd $n$) one may wonder whether on can make this
``integrability'' more explicit. Indeed\footnote{We thank B. Eynard
  for suggesting \eqref{eq:vn.tau}.}
, writing
\begin{align}\label{eq:vn.tau}
  v_n = \frac{\tau_{n+1}\tau_{n-1}}{\tau_n^2}
\end{align}
produces polynomial functions of linearly growing degree (note the gap
between $\tau_2$ and $\tau_3$), which (due to the special nature of
the problem, cp. \eqref{eq:vnx}), are products of 3 successive $u$'s:
\begin{align}
  \tau_n(x)=u_{n-1}(x)u_{n-2}(x)u_{n-3}(x),
\end{align}
where for convenience we list the first few $u$'s, occurring in
\eqref{eq:v.two}/~\eqref{eq:v.three}/~\eqref{eq:v.four}:
\begin{equation}
  \begin{split}
    &u_0(x) = x,\quad
    u_1(x) = \epsilon-x,\quad
    u_2(x) = x^2+x(1+\epsilon)-\epsilon\\
    &u_3(x)=-\epsilon\paraa{4x^2+x(1-2\epsilon)-\epsilon},\quad
    u_4(x) = -\epsilon\paraa{(3-2\epsilon)x^2-2x\epsilon(\epsilon+4)+5\epsilon^2}\\
    &u_5(x) = -3\epsilon\paraa{x^3(1+6\epsilon)+x^2(1+3\epsilon+4\epsilon^2)-x\epsilon(2+3\epsilon+2\epsilon^2)+\epsilon^2(1-\epsilon)}.
  \end{split}
\end{equation}
The conserved quantity (analogue of \eqref{eq:rrr.rec}
resp. \eqref{eq:alphan.recursion}) written in terms of $\tau$'s reads
\begin{align}
  \tau_{n+2}\tau_n\tau_{n-1}^3\tau_{n-2}+\tau_{n+1}^2\tau_{n-1}^2\tau_{n-2}
  -\tau_{n+1}\tau_n^3\paraa{\tau_{n-2}^2+\tau_{n-1}\tau_{n-3}}
  =\epsilon\tau_{n+1}\tau_n^2\tau_{n-1}^2\tau_{n-2}
\end{align}
and the analogue of \eqref{eq:rrr.sec.rec}
\begin{align}
  \tau_{n+2}\tau_n\tau_{n-1}^2+\tau_{n+1}^2\paraa{\tau^2_{n-1}+\tau_n\tau_{n-2}}
  =\epsilon(n+1)\tau_{n+1}\tau_n^2\tau_{n-1},
\end{align}
which is the easiest one to generate the $\tau$'s:
\begin{align*}
  &\tau_{n+2} = \tau_{n+1}\bracketc{\epsilon(n+1)\frac{\tau_n}{\tau_{n-1}}-\frac{\tau_{n+1}}{\tau_n}-\frac{\tau_{n+1}\tau_{n-2}}{\tau_{n-1}^2}}\\
  &r_{n+1} = \epsilon(n+1)r_{n-1}-r_n-\frac{r_nr_{n-1}}{r_{n-2}}\\
  &\tau_n = \bra{0}\Wd W{\Wd}^2W^2\cdots{\Wd}^n W^n\ket{0}
    =v_0\cdots v_{n-1}\tau_{n-1} = r_{n-1}\tau_{n-1}
\end{align*}
Computer calculations show that the case of 2 monomials,
\begin{align}
  Z_1=W^\rho,\quad
  Z_2=W^q\quad p,q\in\integers,\qquad
  W\ket{n}=w_n\ket{n+1}
\end{align}
similarly corresponds to integrable recursion relations.

\section*{Acknowledgments}

\noindent
We would like to thank P. Clarkson, T. Damour, B. Eynard, M. Hynek and
F. Nijhoff for discussions.

\bibliographystyle{alpha}
\bibliography{qms}

\end{document}